\documentclass[12pt]{article}
\usepackage{amsmath, amssymb}
\usepackage{color}

\newcommand{\N}{{\mathbb{N}}}

\newtheorem{theorem}{Theorem}
\newtheorem{corollary}{Corollary}
\newtheorem{lemma}{Lemma}

\newtheorem{proposition}{Proposition}
\newtheorem{definition}{Definition}

\newenvironment{proof}{\noindent\textit{Proof.}}{\hfill $\Box$\par\bigskip}

\title{What is the least number of moves needed to solve the $k$-peg Towers of Hanoi problem?}%
\author{Roberto Demontis \thanks{robdemontis@gmail.com}}%

\date{ }%

\begin{document}
\maketitle

\begin{abstract}
We prove that the solutions to the $k$-peg Tower of Hanoi problem given by Frame
and Stewart are minimal. 
\end{abstract}

This paper solves the problem of finding the least number of moves needed to
transfer a Tower of Hanoi of $n$ disks, from an initial peg to another one of the $k-1$ other
pegs. This problem generalizes a well known puzzle
proposed and solved in \cite{Lucas:Tour} for the case of three pegs. The
generalization to the case of $k$ pegs was proposed in \cite{Stewart:Advanced}, and can
be phrased as follows: ``Given $k$ pegs and $n$ disks initially stacked on one peg in
decreasing order of size (i.e. no larger disk can be on top of a smaller one), how many
moves are needed to transfer the stack of disks from the initial peg to another peg,
assuming that you can move only one disk at a time and it is not allowed to
place a larger disk on top of a smaller disk?''

Two solutions to this problem were proposed in \cite{Frame:Solution} and
\cite{Stewart:Solution} using methods that have been shown to be equivalent in
\cite{Klavzar-Milutinuvic:Frame}. However, as already observed in \cite{Dunkel:Note},
both presumed solutions make use of a special assumption, which up to date is still
unproven, and restricts their proofs of optimality only to algorithms of a certain
scheme. These two models of solution can be however regarded as empirically optimal, as
verified for up to 20 disks in \cite{Houston-Masum:Exploration}.

In this paper we prove 
\begin{theorem} The solutions to the $k$-peg Tower of Hanoi problem given in \cite{Frame:Solution} and
\cite{Stewart:Solution} are minimal.
\end{theorem}

We introduce now some preliminary terminology and notation. We label $n$ disks with the number $1\ldots n$, with the convention that the disk $j$ is \emph{larger} than
the disk $i$ if and only if $j>i$. We use the symbol $\infty$ to indicate a free peg. The triple $\textbf{(j,i,t)}$, with $1\le j<i\le\infty$ and $j<t\le \infty$, denotes that the disk $j$
moves from being on the disk $i$ to be placed on the disk $t$. We say that \textbf{$j\leq n$ 
is freed on a peg}, when it moves for the first time (i.e. when we find for the first time 
the move $(j,j+1,\infty)$ or, in the case $j=n$ the move $(j,\infty,\infty)$). 

\begin{definition}
A sequence of moves is said to be \textbf{demolishing sequence} if:
\begin{enumerate}
	\item it ends its moves with the triple $(n,\infty,\infty)$ and 
	\item the triple $(n,\infty,\infty)$ appears exactly once.
\end{enumerate} 
We can always split a general sequence $S$ solving the Tower of Hanoi problem into a \textbf{demolishing phase} and a\textbf{ reconstructing phase}, we call demolishing phase the moves until the first triple $(n,\infty,\infty)$, and reconstructing phase the remaining moves. 
\end{definition}

Let $H_k(n)$ be the minimum number of moves needed to transfer a Tower of Hanoi of $n$ disks with $k$ pegs. We have:

\begin{lemma}
Let $S$ be a minimal sequence solving Tower of Hanoi problem, then its demolishing phase is composed by $\frac{H_k(n)+1}{2}$ moves.
\end{lemma}
\begin{proof}
We observe the following: suppose we have made the move $(n,\infty,\infty)$, and let this move be
the $l$-th move; for every $u<l$, let $(r,s,y)$ be the $l-u$-th move; then continuing
with the symmetric sequence of moves in which $(r,y,s)$ the $l+u$-th move for every $u<l$, we can
reconstruct the Tower on a peg different from the original one. Thus we may conclude that
in a minimal sequence the triple $(n,\infty,\infty)$ appears exactly once. As a consequence, a minimal sequence is known if its demolishing phase or its reconstructing phase is known. (We note that in general it is not always the case that all minimal sequences have
symmetric demolishing and reconstructing phases, but in any case there are minimal
sequences with this feature). It means that if $S$ is minimal the number of moves of its demolishing phase minus 1 (the move $(n,\infty,\infty)$ must be the same of the number of moves of its reconstructing phase, otherwise we could use the shortest phase to build a sequence shorter than $S$.
As a consequence we have that a demolishing phase of a minimal sequence is composed by $\frac{H_k(n)+1}{2}$ moves.
\end{proof}

As a consequence of Lemma 1, if we want to look for a minimal solution of the $k$-peg Tower of Hanoi problem, we can focus on minimal demolishing phases. Moreover, as we proved that it is always possible to build a minimal sequence having
symmetric demolishing and reconstructing phases, we restrict our attention to sequences having this feature, which we call \textbf{minimal symmetric sequences}.

\begin{definition}
A set of disks arranged on the same peg is said to be a \textbf{stack}. 
\end{definition}

As we have only k pegs, we can build at most k stacks. We can from now restrict our attention for $n>k-1$, as it is easy to calculate the minimum number of moves needed to solve a k-peg Tower of Hanoi problem of $n\leq k-1$ disks.
For minimal demolishing sequences we have the following property:

\begin{proposition}\label{prop:1}
Let us consider a Tower of Hanoi with $n$ disks and one of its minimal demolishing sequence $S$. Suppose that the disks have been arranged on $r\leq k-1$ stacks at the end of $S$. Let
$n,n-1$ and $ j_1< \ldots< j_{r-2}$ be the disks that lie at the bottom
of the $r$ stacks in $S$. 
Then during the demolishing phase no disk $y$, with $y> j_1$, is arranged on the peg on which the disk $j_1$ will be stacked when the move $(n,\infty,\infty)$ is performed.

\end{proposition}

\begin{proof}

Let $S$ be a minimal sequence of moves yielding a demolishing phase. 
 For the sake of a contradiction, let us assume that at some instant during the computation some disk $y$, with $y> j_1$ is stacked on the peg on which $j_1$ will be stacked when the move $(n,\infty,\infty)$ is performed and we assume that $y$ is the last disk bigger than $j_1$ that happens to be stacked on this peg. 
Notice that since $y$ has to leave the peg to make room for $j_1$, there must
be in $S$ a move $(y,\infty,p)$. Since no disk larger than $y$ will ever be stacked on that
peg starting from this instant, we may avoid the move $(y,\infty,p)$ and build a sequence of
moves $S'$ as follows:
\begin{enumerate}
  \item $S'$ coincides with $S$ up to, but not including, the last move of the form
      $(y,\infty,p)$.
  \item Delete the move $(y,\infty,p)$ and all subsequent moves that concern the disk
      $y$.
  \item In $S'$  substitute the moves of the form $(x,z,\infty)$
      going to peg on which $j_1$ will be stacked when the move $(n,\infty,\infty)$ in $S$, with the move $(x,z,y)$.
  \item Iin $S'$  substitute the moves of the form $(x,z,y)$
      placing a disk on the disk $y$, with the move $(x,z,y')$ (where $y'$ is
      the disk, possibly $y'=\infty$, on which lies $y$ in the computation $S$ when a
      move $(x,z,y)$ is performed).
\end{enumerate}
The sequence $S'$ is shorter than $S$, since in $S'$ the move $(y,\infty,p)$ is missing. But
this contradicts our assumption that $S$ is minimal.
\end{proof}

Thanks to the previous Proposition, the following sentences are true:

\begin{corollary}
If we arrange the disks on $k-1$ pegs, there is at least one peg on which no disk bigger than $j_1$ is ever
      stacked. In other words all turns on which a disk bigger than $j_1$ is moved may be simulated by considering the  $(k-1)$ -peg Tower of Hanoi problem.
\end{corollary}

\begin{corollary}
By induction on $k$ we can prove Proposition 1 also for the disks $ j_2< \ldots< j_{r-2}$ .
\end{corollary}

\begin{definition}
Given a sequence of moves $S$ of $n$ disks, we denote with $C_S(j)$ the number of moves
of a disk $j$ in $S$. $C_S(j)$ is said \textbf{cost} of $j$. 
\end{definition}

\begin{corollary}
The cost of each disk lying at the bottom of the stacks, after performing the
     move $(n,\infty,\infty)$ in a minimal demolishing phase, is one.
\end{corollary}
\begin{proof}
To see this, let $x$ be one of these disks. We know that when a move $(x,x+1,\infty)$ is made for
      the first time, $x$ transfers to a free peg; moreover, by the previous Proposition and by Corollary 2
      we know that there is a peg to which no disk bigger than $x$ is ever
      transferred. Therefore, suitably dealing with the disks that are smaller than
      $x$, we can arrange that $x$ is eventually freed
      just on a peg never used by disks larger than $x$: while being on such a peg,
      the disk is of no obstacle to any move and thus need not be moved.
\end{proof}

\begin{definition}
A sequence $S$ is said an \textbf{ideal sequence} of $n$ disks if and only if:
  \begin{enumerate}
	\item it is  a minimal demolishing sequence or
	\item there exists a minimal sequence $S'$ of $m+n$ disks such that every disk bigger than $n$ has been moved using only $k-1$ pegs as described by Corollary 1 with $j_1=n$ and the subsequence until the move $(n,n+1,\infty)$ is $S$.
\end{enumerate}
\end{definition}

\begin{definition}
Define $\textit{S}_k(n)=\{S:S $ is an ideal sequence of $n$ discs using $k$ pegs $\}$ and  
 $X=\{x \in \N: \exists j,n,k$ such that $C_S(j)=x$ for some $S \in \textit{S}_k(n)\}.$ 
\end{definition}
We will consider the ordered set $(X,<)$ with the convention that $x_r<x_{r+1}$ for each $r\geq 1$.  Considering simple  minimal demolishing sequences of k disks we can conclude that
$1, 2\in X$. As example $$(1,2,\infty)(2,3,\infty)\ldots (k-1,k,\infty)(1,\infty,2)(k,\infty,\infty).$$

\begin{definition}
 For all $x_i\in X$ let $M_k(x_i)=\sup_n \max_{S\in \textit{S}_k(n)}|\{ j:C_S(j)=x_i\}|.$ 
\end{definition}

\begin{lemma}
Fix an integer $n\geq k-1$. Let i be the maximum integer such that $\sum_{t=1}^{i-1}M_k(x_t)\leq n$ and set $T= n- \sum_{t=1}^{i-1}M_k(x_t)$, eventually $T=0$. Then,
\begin{eqnarray}
H_k(n)\geq 1+2*(k-2) +\sum_{t=2}^{i-1}M_k(x_t)2x_t+ T2x_{i}.
\end{eqnarray}
\end{lemma}
\begin{proof}
  By Proposition 1 we know that $M_k(1)=k-1$.
Moreover by Lemma 1 we know that only one disk costs 1 move  and as a result at most $k-2$ disks costs 2 moves.
If $t>1$ by definition of $M_k(x_t)$ and considering minimal symmetric sequences, we have that $M_k(x_t)$ is the maximum number of disks that can be moved at cost $2x_t$. We are not sure that every $x_i\in X$ is a cost in minimal demolishing phase of a minimal symmetric sequence but adding these costs we are decreasing our valuation and as a consequence we obtain (1).
\end{proof}

\begin{definition}
For all $x_i\in X$ let $L_k(x_i)=\sum_{t=1}^i M_k(x_t)$.
\end{definition}
Thus $L_k(x_i)$ is an upper
bound on the maximum  number of disks that we can move in a minimal demolishing sequence at a cost less or equal to $x_i$ considering all the possible costs in the set of all ideal sequences. 

\begin{definition}
Define $\textit{SS}_k(n)=\{T:T $is a minimal symmetric sequence of $n$ discs on $k$ pegs$\} $
\end{definition}
\begin{definition}
Let $I_k(x_i)=\sup_n \max_{T\in \textit{SS}_k(n)}|\{ j:C_T(j)\leq x_i\}|.$
\end{definition}

In other words $I_k(x_i)$ is the maximum number of disks that we can transfer from a peg to another peg at a cost smaller or equal to $x_i$ in the set $\textit{SS}_k(n)$. 

\begin{definition}
Define $m_k(1)=I_k(1)$ and for $i\geq 2$ define $m_k(x_i)=I_k(x_i)-I_k(x_{i-1})$.
\end{definition}

\begin{lemma}
$I_k(x_i)\leq L_k(x_{i-1})$
\end{lemma}
\begin{proof}
  For all $i\geq 2$ if the disks have a cost smaller or equal to $x_i$ in a minimal symmetric sequence $S$, they must have cost $x_{i}/2$ in the demolishing sequence of $S$. By definition we have that the demolishing sequence of $\textit{SS}_k(n) $ are subset of $ \textit{S}_k(n)$, as a consequence $x_{i}/2\in X$, but also $x_i \in X$ then $x_{i}/2\leq x_{i-1}$. In conclusion we have  $ I_k(x_i)\leq L_k(x_{i}/2) \leq L_k(x_{i-1})$. 
\end{proof}

 We note that $L_k(x_{i-1})$ is the maximum value that $I_k(x_i)$ could be. Note that if we increase the number of disks that we can move at a smaller cost we decrease the lenght of a sequence. As a consenquence it seems to make sense to consider the case $I_k(x_i)= L_k(x_{i-1})$ and $m_k(x_i)=M_k(x_{i-1})$. More formally we have the following

\begin{lemma}
Fix an integer $n\geq k-1$ and suppose for all $ i\geq3$ $m_k(x_i)=M_k(x_{i-1})$.
 Let i be the maximum integer such that $\sum_{t=1}^{i}m_k(x_t)\leq n$ and set $T'= n- \sum_{t=1}^{i}m_k(x_t)$. Then
\begin{eqnarray}
H_k(n)\geq1+2*(k-2) +\sum_{t=3}^{i}m_k(x_t)2x_{t-1}+ T'2x_{i}.
\end{eqnarray}
 \end{lemma}
\begin{proof}
If $m_k(x_i)=M_k(x_{i-1})$ as a result $(2)=(1)$, then by Lemma 2  we have $H_k(n)\geq (2)$.
\end{proof}

Now, as a consequence of Lemma 4,  we have that when we suppose $m_k(x_i)=M_k(x_{i-1})$ the telescopic sum defined in (2) is smaller than $H_k(n)$ for all $n$. Then, as we are looking for a minimum of $H_k(n)$, we can estimate this telescopic sum  just  in the case of $m_k(x_i)=M_k(x_{i-1})$ for all $i\geq 3$.

 \begin{lemma}
If for all $i\geq 2$ $I_k(x_i)= L_k(x_{i-1})$, then $I_k(x_t)\leq {k-3+t\choose k-2}$.
\end{lemma}
\begin{proof}
 By induction on $k$ and $i$, it is trivial to prove it for $k=3$ and $ i=2$, suppose this is true for $ k-1$ and for all $p<i$.  Using Proposition 1, we know that in a minimal sequence $S$ we can always find a disk $j_k$ such that every disk bigger than $j_k$ has been moved using only $k-1$ pegs. Then we can split $I_k(x_i)$  between $I_{k-1}(x_i)$ and the number of disks smaller than $j_k+1$. Let $N$ be this number. 

Moreover the demolishing sequence $S$ just until the first move of $j_k$ is an ideal sequence, as a consequence for all disk $j\leq j_k$  $ C_S(j)\leq x_{i-2}$ and $ N\leq L_k(x_{i-2})$. As a result 
 $$I_k(x_i)\leq I_{k-1}(x_i)+L_k(x_{i-2})=I_{k-1}(x_i)+I_k(x_{i-1}).$$ Then by induction  we have $I_k(x_i)\leq {k-4+i\choose k-3}+{k-4+i\choose k-2}={k-3+i\choose k-2}$.
\end{proof}

As a consequence the maximum number of disk that we can move at cost $x_i$ is $m_k(x_i)=I_k(x_i)-I_k(x_{i-1})= {k-4+i\choose k-3}$, this implies that $m_k(x_i)$ must be different to 0.

As a result if $m_k(x_i)$ is different to 0, we prove the following Lemmas:

\begin{lemma}
 For all $i$  $x_i\geq 2x_{i-1}$.
\end{lemma}
\begin{proof}
For the sake of a contradiction, let $x_y\in X$ the smallest number such that $x_y< 2x_{y-1}$. 
By Lemma 1 we know $I_k(x_y)$ is equal to the number of disks that can be freed at the maximum cost of $x_y/2$. As we are supposing that $x_y\in X$ is the smallest number such that  $x_y< 2x_{y-1}$, we must have $x_{y-1}\geq 2x_{y-2}$,  as a consequence we have $x_{y-2}\leq x_{y-1}/2<x_y/2<x_{y-1}<x_y$. Then we claim $I_k(x_y)=I_k(x_{y-1})$, for if not then there would be a minimal demolishing sequence $S$ and a disc $j$ such that  $x_{y-1}/2< C_s(j)\leq x_y/2$. But this would mean $x_{y-2}< C_s(j)<x_{y-1}$ which would be a contradiction. Moreover if $I_k(x_y)=I_k(x_{y-1})$ then $m_k(x_y)= 0$. We find a  contradiction and we must have $x_i\geq 2x_{i-1}$. 
\end{proof}

Now we know $x_1=1$, this means that under the assumptions of Lemma 5  $x_i\geq 2^{i-1}$. As we are looking to estimate the minimum for the telescopic sum (2) we can suppose $x_i=2^{i-1}$.

In conclusion we kwon that the minimum cost we can suppose is $x_i=2^{i-1}$ and the maximum number of disk that we can move at cost $x_i$ is $m_k(x_i)={k-4+i\choose k-3}$, then we can write (2) as follows: 
\begin{eqnarray}
\sum_{t=1}^i 2^{t-1}{k-4+t\choose k-3}+ 2^i\left(n-{k-4+i\choose k-3}\right)
\end{eqnarray}

 with ${k-4+i+1\choose k-3}>n\geq {k-4+i\choose k-3}.$

\begin{lemma}
$H_k(n)\geq (3)$.
\end{lemma}
\begin{proof}
By Lemma  4, 5, 6.
\end{proof}
 As (3) is exactly the number of moves needed to solve the Tower of Hanoi problem given by the methods proposed in \cite{Frame:Solution} and
\cite{Stewart:Solution} for the case of $k$ pegs as proved in \cite{Klavzar-Milutinuvic:Frame}, we can conclude that these methods produce minimal sequences and in conclusion we proved Theorem 1.


\begin{thebibliography}{1}

\bibitem{Dunkel:Note} O.~Dunkel.
\newblock Editorial note.
\newblock {\em Amer. Math. Monthly}, 48:219, 1941.

\bibitem{Frame:Solution} J.~S. Frame.
\newblock Solution to advanced problem 3918.
\newblock {\em Amer. Math. Monthly}, 48:216--217, 1941.

\bibitem{Houston-Masum:Exploration} B.~Houston and H.~Masum.
\newblock Explorations in 4-peg {T}ower of {H}anoi.
\newblock Technical Report TR-04-10, Carleton University, Ottawa, Canada, 2004.

\bibitem{Klavzar-Milutinuvic:Frame} Sandi Klav\v{z}ar, Uro\v{s} Milutinovi\'{c}, and
    Ciril Petr.
\newblock On the {F}rame-{S}tewart algorithm for the multi-peg {T}ower of
  {H}anoi problem.
\newblock {\em Discrete Appl. Math.}, 120(1-3):141--157, 2002.

\bibitem{Lucas:Tour} E.~Lucas.
\newblock La {T}our d'{H}anoi, jeu de calcul.
\newblock {\em Science and Nature}, 8(1):127--128, 1884.

\bibitem{Stewart:Advanced} B.~M. Stewart.
\newblock Advanced problem 3918.
\newblock {\em Amer. Math. Monthly}, 46:363--364, 1939.

\bibitem{Stewart:Solution} B.~M. Stewart.
\newblock Solution to advanced problem 3918.
\newblock {\em Amer. Math. Monthly}, 48:219, 1941.

\end{thebibliography}

\end{document}